\documentclass[11pt]{amsart}
\usepackage{tikz-cd}
\usepackage{graphicx}
\usepackage{amsmath}
\usepackage{amsfonts}
\usepackage{amssymb}
\newtheorem{thm}{Theorem}[section]
\theoremstyle{definition}

\newtheorem{defn}[thm]{Definition}
\newtheorem{remark}[thm]{Remark}
\newtheorem{exam}[thm]{Example}

\usepackage{xcolor}

\DeclareMathOperator{\Tr}{Tr}

\begin{document}

\title[]{Axiomatic Approach to Quantum Superchannels}

\author[P.~C.~Daly]{Pádraig C.~Daly}
\address{Institute for Quantum Computing and Department of Pure Mathematics, University of Waterloo,
Waterloo, ON, Canada  N2L 3G1}
\email{pcdaly@uwaterloo.ca}

\maketitle
\mbox{}\vspace*{-1\baselineskip}
\begin{abstract} Quantum superchannels are maps whose input and output are quantum channels. Rather than taking the domain to be the space of all linear maps we motivate and define superchannels on the operator system spanned by quantum channels. Extension theorems for completely positive maps allow us to apply the characterisation theorem for superchannels to this smaller set of maps. These extensions are non unique, showing two different superchannels act the same on all input quantum channels, and so this new definition on the smaller domain captures more precisely the action of superchannels as transformations between quantum channels. The non uniqueness can affect the auxilliary dimension needed for the characterisation as well as the tensor product of the superchannels.

\end{abstract}

\section{Introduction}
Quantum channels are a fundamental object studied in quantum information \cite{Holevo+2012, nielsen_chuang_2010}. Defined as completely positive trace-preserving (CPTP) maps between operators on Hilbert spaces, they map quantum states to quantum states. Since quantum states are positive operators with trace one, the natural domain and range of quantum channels is taken to be the ideal of trace class operators inside the space $\mathcal{B}(\mathcal{H})$ of bounded operators on a Hilbert space. Quantum superchannels are one step up from this, transformations between quantum channels.

Quantum superchannels were introduced in \cite{original} to describe the most general transformation of quantum channels, and have been used as a model of quantum circuit boards with the ability to replace quantum channels \cite{chiribella2008quantum}. Recent work has used superchannels to define the entropy of quantum channels \cite{Gour}, and study dynamical resource theories such as entanglement \cite{gour2021entanglement}, magic \cite{wang2019quantifying}, and coherence \cite{liu2020operational}.  Concepts from quantum channels, such as entanglement-breaking and dephasing, have been extended to the superchannel case \cite{chen2020entanglement, puchala2021dephasing} to understand how these properties can be introduced as channels change.

In \cite{original} and \cite{Gour} the domain of these superchannels is taken to be the set of quantum operations which are completely positive trace non-increasing maps between operators. In finite dimensions the span of these maps gives all of the linear maps on spaces of operators. A characterisation of all superchannels in these papers describes them as the action of two ordinary quantum channels, a ``pre-processing" and ``post-processing" channel.

In this paper we propose a different definition of superchannels. In particular, we take the domain to be the span of quantum channels which in general is not all linear maps between operators. Considering the Choi matrices \cite{CHOI1975285} of such maps allows us to define an operator system and make use of Stinespring's theorem \cite{stinespring}. Arveson's extension theorem \cite[Theorem 1.2.3]{Arveson1969OnSO} allows the same characterisation of the more general superchannels to apply to this smaller class of maps.

We then show that these extensions are non unique, meaning different extensions give superchannels whose action on quantum channels is the same. This shows that the usual definition of superchannel results in different maps which have the same effect on quantum channels. This provides evidence that this new definition is more natural as a description of maps on channels.

Consequences of the non uniqueness of these extensions are then explored. It is shown there are no TP extensions, and that the tensor product can be affected. The extreme points of the set of extensions is examined in a generalisation of a theorem by Choi \cite{CHOI1975285}.

\section{Preliminaries}
For notation let $M_n$ be the space of $n\times n$ matrices over the complex numbers and let $B(\mathcal{H})$ be the space of bounded operators a Hilbert space $\mathcal{H}$. $M_n (\mathcal{H}) = M_n \otimes \mathcal{H}$ is the space of $n\times n$ matrices with entries in $\mathcal{H}$.

An \textit{operator system} $\mathcal{S}$ is a subspace of a unital $C^*$-algebra which contains the unit and is self-adjoint; i.e., $\mathcal{S}=\mathcal{S}^*=\{a^*:a\in \mathcal{S}\}.$ If $\mathcal{B}$ is a $C^*$-algebra and $\phi : \mathcal{S}\xrightarrow[]{} \mathcal{B}$ is a linear map then $\phi$ is \textit{positive} if it maps positive elements to positive elements. Define $\phi_n : M_n(\mathcal{S})\xrightarrow[]{}M_n(\mathcal{B})$ by $\phi_n ((a_{i,j}))=(\phi(a_{i,j}))$. That is, $\phi_n = \phi \otimes \text{id}_n$ where $\text{id}_n$ is the identity map on $M_n$. Call $\phi$ \textit{completely positive} (CP) if $\phi_n$ is positive for all $n$.

The key idea with operator systems is that they are defined by their matrix order; i.e., the cones of positive elements in $M_n(\mathcal{S})$ for each $n$. As subspaces of $C^*$-algebras this is defined using bounded operators on a Hilbert space. To show two spaces define the same operator system it is necessary to show their matrix orders are the same and for this we use a complete order isomorphism. A linear map $\phi : \mathcal{S}\xrightarrow[]{}\mathcal{T}$ between operator systems is a \textit{complete order isomorphism} if it is bijective and both $\phi$ and $\phi^{-1}$ are completely positive. See \cite{vern} for more on operator systems.

Stinespring's dilation theorem says that if $\mathcal{A}$ is a unital $C^*$-algebra and $\phi : \mathcal{A}\xrightarrow[]{}B(\mathcal{H})$ is a completely positive map, then there exists a Hilbert space $\mathcal{K}$, a unital *-homomorphism $\pi : \mathcal{A}\xrightarrow[]{}B(\mathcal{K})$, and a bounded operator $V:\mathcal{H}\xrightarrow[]{}\mathcal{K}$ such that $\phi (a) = V^*\pi (a) V$. In the finite dimensional case, for a CP map $\phi :M_d \xrightarrow[]{} M_r$ there are a collection of operators $V_i : \mathbb{C}^r \xrightarrow[]{}\mathbb{C}^d$ called the \textit{Kraus operators} such that the action of $\phi$ is given by
\begin{equation*}
    \phi (X) = \sum_{i}^m V_i^* X V_i.
\end{equation*}

Let $E_{i,j}$, $1\leq i,j\leq d$ denote the matrix units in $M_d$. Now any linear map is determined by its action on basis elements so for a linear $L:M_d \xrightarrow[]{} M_r$ we get a vector space isomorphism from $\mathcal{L}(M_d,M_r)$ onto $M_d (M_r)$ via $L \mapsto (L(E_{i,j}))$. The matrix $C_L := (L(E_{i,j}))$ is called the \textit{Choi matrix} or \textit{Choi-Jamiołkowski matrix} of the map. Choi's theorem says that a linear map $\phi : M_d \xrightarrow[]{} M_r$ is completely positive if and only if $C_{\phi} \geq 0$ in $M_d (M_r)$.

We will consider finite-dimensional quantum channels which are defined as linear, completely positive trace-preserving (CPTP) maps $\phi : M_d \xrightarrow[]{}M_r$. Trace-preserving means $\Tr(\phi(X))=\Tr(X)$ for all $X \in M_d$.

\section{Defining Quantum Superchannels}

We are now ready to define the space of quantum channels. Note that the Choi matrix of a quantum channel is a block matrix where the diagonal blocks each have trace one and the off diagonal blocks have trace zero.
\begin{defn} Given positive integers $d,r \geq 1$ define $SCPTP (d,r) := \text{span}\{\phi | \phi: M_d\xrightarrow[]{}M_r \text{ is a CPTP map}\} \subset \mathcal{L}(M_d, M_r)$. Also define $S(d,r)\subset M_d (M_r)$ to be the set of block matrices $(P_{i,j})$ such that for all $1\leq i,j \leq d$,  $\Tr (P_{i,i}) = \Tr (P_{j,j})$ and for $i\neq j$ $\Tr (P_{i,j}) =0$.
\end{defn}
There is a natural way to define a matrix order on the space $SCPTP(d,r)$: for an $n\times n$ matrix $(\phi_{i,j})$ of maps, with each $\phi_{i,j}\in SCPTP (d,r)$, define $\Phi : M_d \xrightarrow[]{} M_n(M_r)$ by $\Phi (x) = (\phi_{i,j}(x))$. Then $(\phi_{i,j})\geq 0$ if and only if $\Phi$ is completely positive.

Note in the case $d=1$ we have $M_1 = \mathbb{C}$ and since any such linear map is defined by its value at $1$ we have  an isomorphism $ \mathcal{L}(M_1,M_r) \cong M_r$ via $\phi \mapsto \phi (1)$. The positive matrices span the whole space so in fact $SCPTP (1,r) \cong M_r$. With the order on $SCPTP (d,r)$ that we just defined this is a complete order isomorphism. Similarly $S(1,r) \subset M_1 (M_r) = M_r$ and since there is just one block $P_{1,1}$ with no restriction we get all the $r\times r$ matrices in $S(1,r)$. Thus $SCPTP (1,r)$ is order isomorphic to  $S(1,r)$.

\begin{thm}\label{2.2}
$S(d,r)$ is an operator system and is completely order isomorphic to $SCPTP (d,r)$ via the Choi map
\begin{equation*}
    \begin{split}
        R: SCPTP(d,r) &\xrightarrow[]{} S(d,r) \\
        \phi &\mapsto C_{\phi}.
    \end{split}
\end{equation*}
\end{thm}
\begin{proof}
It is clear that $S(d.r)$ contains the identity matrix and the linearity of the trace ensures it is a subspace. If $X=(P_{i,j})$ is a block matrix then the adjoint is $X^*=(P_{j,i}^*)$ and for any block $\Tr (P) = \overline{\Tr (P^*)}$. This implies that $S(d, r)$ is self-adjoint and hence an operator system.

Next show $\phi \mapsto C_{\phi}$ is an isomorphism between $SCPTP (d,r)$ and $S(d,r)$. This is the correct range because the Choi matrix for a quantum channel is in $S(d,r)$. It is injective because any linear map is defined by its Choi matrix. To prove surjectivity we use the fact about operator systems that any $X \in S(d,r)$ can be written in terms of four positive matrices $P_i \in S(d,r)$, $1\leq i \leq 4$, as
\begin{equation*}
    X= (P_1 - P_2 ) + i(P_3 - P_4 ).
\end{equation*}
As they are positive $\Tr (P_i) = 0$ only if $P_i = 0$. Thus we can scale each $P_i$ by a factor $1/\Tr (P_i)$ to make it into a Choi matrix associated with a CPTP map. This proves any $X\in S(d,r)$ is in the span of Choi matrices of quantum channels.

Finally we show it is a complete order isomorphism. For a matrix of maps in $SCPTP (d,r)$ the condition to be positive is
\begin{equation*}
    \left(\phi_{i,j}\right)_{i,j}\geq0 \iff \Phi \text{ CP} \iff \left(\Phi (E_{k,l})\right)_{k,l}\geq 0  \iff \left(\left(\phi_{i,j} \left(E_{k,l}\right)\right)_{i,j}\right)_{k,l}\geq 0.
\end{equation*}
The corresponding matrix of Choi matrices can be written
 $$
 \left(C_{\phi_{i,j}}\right)_{i,j} = \left(\left(\phi_{i,j}\left(E_{k,l}\right)\right)_{k,l}\right)_{i,j}
 $$
To conclude note that the shuffle which maps $M_m (M_n (\mathcal{A}))$ to $M_n (M_m (\mathcal{A}))$ is a $*$-isomorphism and hence preserves positivity.

\end{proof}

\begin{remark}
A tensor product of linear maps gives a map on the tensor product of the spaces so there is an inclusion

$$SCPTP (d_1,r_1) \otimes SCPTP (d_2,d_2) \subseteq SCPTP (d_1d_2,r_1r_2 ).$$
We can show this is generally a strict inclusion and the spaces are not equal. The description of $S(d,r)$ allows us to do a dimension count giving the dimension of $SCPTP(d,r)$ as $d^2r^2 -d^2 +1$. So for the tensor product space we have dimension $(d_1^2r_1^2 -d_1^2 +1)(d_2^2r_2^2 - d_2^2 +1)$ but for the space on the right we have dimension $d_1^2d_2^2r_1^2r_2^2 -d_1^2d_2^2 +1$ which is generally larger. The difference is
$$d_1^2d_2^2(r_1^2+r_2^2 -2) -d_1^2 (r_1^2-1)-d_2^2(r_2^2-1)$$
and this is non-negative. We endow this tensor space with an order by regarding its elements as elements of the operator system $SCPTP (d_1d_2,r_1r_2)$.
\end{remark}

To motivate the definition of quantum superchannels it is worth recalling some reasons behind the definition for ordinary quantum channels. Two simple requirements were that channels be linear maps that take quantum states to quantum states. This gives the trace-preserving condition. The requirements that quantum systems combine using tensor products, and that the identity map is a valid channel is what implies the completely positive condition. For superchannels we similarly require they be linear maps which takes channels to channels, and that the tensor of any two superchannels is again a superchannel on the combined space.

\begin{defn}
Given two spaces of quantum channels $SCPTP (d_i,r_i)$, $i=1,2$, a \textit{QSC} is a linear map $\Gamma : SCPTP(d_1,r_1) \xrightarrow[]{} SCPTP(d_2,r_2)$
which satisfies
\begin{enumerate}
    \item if $\phi$ is CPTP then $\Gamma (\phi )$ is CPTP
    \item given any other dimensions $d_3, r_3 \in \mathbb{N}$ and the identity map $\text{id}_{d_3, r_3} :SCPTP (d_3,r_3)\xrightarrow[]{} SCPTP (d_3,r_3)$ then $\Gamma \otimes \text{id}_{d_3, r_3} : SCPTP (d_1,r_1) \otimes SCPTP (d_3,r_3) \xrightarrow[]{} SCPTP (d_2,r_2 ) \otimes SCPTP (d_3,r_3)$ sends CP maps to CP maps.
\end{enumerate}
\end{defn}

Let $R_i: SCPTP (d_i,r_i) \xrightarrow[]{} S(d_i,r_i)$ be the complete order isomorphism sending $\phi$ to $C_{\phi}$. If $\Gamma$ is a QSC it induces a map $\widetilde{\Gamma}: S(d_1,r_1)\xrightarrow[]{} S(d_2,r_2)$ via
$$\widetilde{\Gamma} = R_2 \circ \Gamma \circ R_1^{-1}.$$
Explicitly this acts as $\widetilde{\Gamma} (C_{\phi}) = C_{\Gamma (\phi)}$. It is useful to study this map because the properties of QSC's implies that $\widetilde{\Gamma}$ is completely positive. Note that by Choi's theorem $\widetilde{\Gamma}$ sends positive matrices to positive matrices.

\begin{thm}
If $\Gamma :SCPTP (d_1,r_1)\xrightarrow[]{} SCPTP (d_2,r_2)$ preserves CPTP maps then it is a QSC if and only if $\widetilde{\Gamma}$ is completely positive.
\end{thm}
\begin{proof}
Recall that $S(1,n) = M_n$. Take the identity map on $SCPTP (1,n)$ which has as its induced map on $S(1,n)$ the identity on $M_n$. Let $C_{\phi} \in S(d_1, r_1) \otimes M_n$ be a positive matrix. If $R_i: SCPTP (d_i,r_i) \xrightarrow[]{} S(d_i,r_i)$ are the Choi isomorphisms we can write

$$\widetilde{\Gamma} \otimes \text{id}_n (C_{\phi}) =( R_2 \otimes R_3) (\Gamma \otimes \text{id}_{1,n})(\phi).$$ Then the second property of QSC's implies that $\widetilde{\Gamma} \otimes \text{id}_n$ sends positive matrices to positive matrices for all $n$. Thus $\widetilde{\Gamma}$ is a completely positive map.

For the converse, suppose $\widetilde{\Gamma}$ is completely positive and note that $\text{id}_{d_3,r_3}: SCPTP (d_3,r_3)\xrightarrow[]{} SCPTP (d_3,r_3)$ is a QSC, and thus has a completely positive induced map. For any CP map $\phi$ we have $$\Gamma \otimes \text{id}_{d_3,r_3} (\phi) = (R_2^{-1} \otimes R_{3}^{-1}) (\widetilde{\Gamma} \otimes \widetilde{\text{id}_{d_3,r_3}}) (C_{\phi})$$
and this is a CP map since $\widetilde{\Gamma}$ and $\widetilde{\text{id}_{d_3,r_3}}$ are both completely positive so their tensor is a positive map.
\end{proof}
Thus the second property in the definition of QSC can be replaced by the requirement that QSC's be completely positive. 

Elements of $SCPTP (d,r)$ scale the trace of density matrices and this scaling factor is preserved by QSCs. Consider a QSC $\Gamma :SCPTP (d_1,r_1)\xrightarrow[]{}SCPTP(d_2,r_2)$. Suppose $\phi \in SCPTP(d_1,r_1)$ satisfies $\Tr \phi(X) = c\Tr X$ some constant $c$ and that $\Tr \Gamma (\phi)(Y) = k\Tr Y$ some constant $k$. Decompose $\phi$ as a span of quantum channels
$$\phi = \sum_i c_i \phi_i$$
and use the trace condition on $\phi$ to see $$\sum_i c_i = c.$$
Now since $\Gamma$ is linear and sends TP maps to TP maps we get
$$k\Tr Y= \Tr \Gamma (\phi )(Y) = \sum_i c_i \Tr \Gamma (\phi_i )(Y) = c \Tr Y$$
and so $k=c$.
\begin{remark}
A tensor product of two QSCs $\Gamma_1 \otimes \Gamma_2 : SCPTP (d_1,r_1) \otimes SCPTP (d_3,r_3) \xrightarrow[]{} SCPTP (d_2,r_2) \otimes SCPTP (d_4,r_4)$ will send CPTP maps in the domain to CPTP maps in the range, but it is not a QSC as its domain is not $SCPTP (d_1d_3, r_1r_3).$
\end{remark}

\begin{remark}[QSC vs quantum superchannel]
In \cite{original} and \cite{Gour} the definition of superchannel used the space of all linear maps as its domain and range. In particular a \textit{quantum superchannel} is a linear map $S: \mathcal{L}(M_{d_1},M_{r_1}) \xrightarrow[]{} \mathcal{L}(M_{d_2},M_{r_2})$ which satisfies 
\begin{itemize}
    \item CP preserving: $S$ sends CP maps to CP maps
    \item Completely CP preserving: For any $d,r$ if $\text{id}_{d,r}$ is the identity map acting on $\mathcal{L}(M_d,M_r)$ then $S \otimes \text{id}_{d,r}$ is CP preserving
    \item TP preserving: $S$ sends TP maps to TP maps 
\end{itemize}
\end{remark}

Superchannels and QSC's are defined in similar ways, although on a different space of maps. Since superchannels use the whole vector space of linear maps, tools such as the Choi matrix can be applied to them. This is not the case for QSC's since the space spanned by Choi matrices of quantum channels doesn't contain the standard matrix units, so the Choi matrix cannot be defined. The next theorem allows us to extend QSC's and treat them as restrictions of superchannels.
\begin{thm}\label{2.7}
Every QSC extends to a quantum superchannel.
\end{thm}
\begin{proof}
Let $\Gamma$ be a QSC. \textit{Arveson's extension theorem} says that given a $C^*$-algebra $\mathcal{A}$ containing an operator system $S$ then if $\phi :S \xrightarrow[]{}B(\mathcal{H})$ is a completely positive map there is a completely positive map $\psi : \mathcal{A}\xrightarrow[]{}B(\mathcal{H})$ extending $\phi$. Since $\widetilde{\Gamma}$ is CP it has a CP extension with domain all of $M_{d_1}(M_{r_1})$. Call this extension $\widetilde{S}$.

Define the matrix order on $\mathcal{L}(M_d,M_r)$ in the same way as for $SCPTP(d,r)$ and a similar proof to Theorem \ref{2.2} shows the Choi isomorphism $\phi \mapsto C_{\phi}$ is also a complete order isomorphism between $\mathcal{L}(M_d,M_r)$ and $M_d(M_r)$. Thus $\widetilde{S}$ corresponds to a map $S$ which is an extension of $\Gamma$. We will show $S$ is a quantum superchannel.

Any TP map $f\in \mathcal{L}(M_{d_1},M_{r_1})$ will have a Choi matrix that has trace one on the diagonal blocks and trace zero on the off diagonal blocks. Thus $C_f \in S(d_1,r_1)$ and so we can write $f$ as a linear combination of CPTP maps. Using the linearity of $S$ we can see that $S(f)$ is a TP map.

To see completely CP preserving take a matrix of CP maps $(\phi_{i,j})$. Then $(C_{\phi_{i,j}})$ is a matrix of positive matrices and since $\widetilde{S}$ is completely positive $\widetilde{S}^{(n)}$ maps it to another matrix of positive matrices.
\end{proof}

\begin{remark}
The extension of a QSC is not unique. For example, let $d_1=2$, let $r_1$ be arbitrary size, and let $d_2=r_2=1$. Define $\widetilde{\Gamma}_1, \widetilde{\Gamma}_2: M_2 (M_{r_1})\xrightarrow[]{} M_1 (M_1)$ via

\begin{equation}\label{2lads}
    \begin{split}
        \widetilde{\Gamma}_1 \left( \begin{pmatrix}
\phi (E_{11}) & \phi (E_{12}) \\
\phi (E_{21}) & \phi (E_{22}) 
\end{pmatrix} \right) &= \Tr (\phi (E_{11})), \\
        \widetilde{\Gamma}_2 \left( \begin{pmatrix}
\phi (E_{11}) & \phi (E_{12}) \\
\phi (E_{21}) & \phi (E_{22}) 
\end{pmatrix}  \right) &= \Tr (\phi (E_{22})).
    \end{split}
\end{equation}
These are different maps in general but are identical when restricted to the space of quantum channels $S(2,r_1)$. They are easily seen to be linear maps which send CPTP maps to $1$. To see that they are completely positive take $V_1= \begin{pmatrix}I_{r_1} \\ 0 \end{pmatrix}$, $V_2 = \begin{pmatrix}0\\ I_{r_1} \end{pmatrix}$ then
$$\widetilde{\Gamma}_i(C_{\phi}) = \Tr (V_i^* C_{\phi}V_i).$$
\end{remark}

\begin{remark}
 Define the \textit{depolarizing channel} $\Delta_1:M_{d_1}\to M_{r_1}$
 $$\Delta_1 (\rho)=\frac{Tr(\rho)}{r_1}1$$
 and similarly $\Delta_2:M_{d_2}\to M_{r_2}$. Then the Choi matrices are $C_{\Delta_1}=\frac{1}{r_1}1=\frac{1}{r_1}C_{r_1\Delta_1}$. For $\widetilde{\Gamma}$ to be unital we require
 
 $$\widetilde{\Gamma}(1)=\widetilde{\Gamma}(C_{r_1\Delta_1}) = 1=C_{r_2\Delta_2}$$
 but since $\widetilde{\Gamma}(C_{r_1\Delta_1})=C_{\Gamma(r_1\Delta_1)}$ unital is equivalent to requiring $\Gamma (r_1 \Delta_1)=r_2\Delta_2$. Thus the depolarizing channel is the \textit{order unit} of the operator system $SCPTP(d,r)$, see \cite{vern}.
\end{remark}

\subsection{QSC with no TP extension}
A QSC is defined by a CP map $\widetilde{\Gamma}: S(d_1,r_1) \xrightarrow[]{} S(d_2,r_2)$ which sends block matrices of trace $\lambda d_1$ to block matrices of trace $\lambda d_2$ where $\lambda$ is the trace scaling factor of the linear map associated with the block matrix (with $\lambda =1$ for CPTP maps and their Choi matrix). Thus the map 
$$\frac{d_1}{d_2}\widetilde{\Gamma}$$
is a CPTP map. If we extend $\widetilde{\Gamma}$ to a superchannel $\widetilde{S}: M_{d_1}(M_{r_1}) \xrightarrow[]{} M_{d_2}(M_{r_2})$ then in general it is not the case that $\widetilde{S}$ is TP.

Consider a map $M_2(M_2) \xrightarrow[]{} M_2 (M_2)$
defined by
\begin{equation*}
\begin{split}
E_{11} \mapsto \text{Diag}(a_1,a_2,a_3,a_4) &=A \\
E_{22}\mapsto \text{Diag}(b_1,b_2,b_3,b_4) &= B \\
E_{33}\mapsto \text{Diag} (c_1,c_2,c_3,c_4) &=C \\
E_{44}\mapsto \text{Diag}(d_1,d_2,d_3,d_4) &= D 
\end{split}
\end{equation*}
and all other standard basis matrices get sent to $0$.

Since $E_{11} +E_{33}$, $E_{11}+E_{44}$, $E_{22}+E_{33}$, and $E_{22}+E_{44}$ are in $S(2,2)$ for this map to restrict to give a QSC we require $A+C$, $A+D$, $B+C$, and $B+D$ to be in $S(2,2)$ and have the same trace (since $\frac{d_1}{d_2}=1$) i.e. they must have trace $2$ and both diagonal blocks each have trace $1$. In other words,

\begin{equation*}
    \begin{split}
        a_1 +c_1 +a_2 +c_2 &= 1 \\
        a_3 +c_3+ a_4 + c_4 &=1 \\
        a_1 + d_1 +a_2 +d_2 &=1 \\
        a_3 +d_3 +a_4 +d_4 &=1
    \end{split}
\end{equation*}
and similarly with $b_i$ replacing $a_i$. This implies $a_1 +a_2 = b_1 +b_2$ and $a_3 +a_4 = b_3 +b_4$.

For this to be a trace-preserving map we require $\sum_{i} a_i = \sum_i b_i = \sum_i c_i = \sum_i d_i =1$.

So for a particular choice of $A, B, C, D$ which give a QSC with no TP extension consider

$$A= \begin{bmatrix} 
\frac{1}{2} & 0 & 0 & 0 \\
0 & \frac{1}{2} & 0 & 0 \\
0 & 0 & 1 & 0 \\
0 & 0 & 0 & 0
\end{bmatrix}$$

$$ B= \begin{bmatrix} 
1 & 0 & 0 & 0 \\
0 & 0 & 0 & 0 \\
0 & 0 & 0 & 0 \\
0 & 0 & 0 & 1
\end{bmatrix}$$

$$ C= \begin{bmatrix} 
0 & 0 & 0 & 0 \\
0 & 0 & 0 & 0 \\
0 & 0 & 0 & 0 \\
0 & 0 & 0 & 0
\end{bmatrix}.$$
This QSC is defined by the matrices $A+C$ and $B+C$. Any choice of $a_i$ and $c_i$ must satisfy $a_1 +c_1 = \frac{1}{2}$, $a_2+c_2 = \frac{1}{2}$,  $a_3+c_3 =1$, and $a_4 +c_4 = 0$ to be the same QSC.

However for an extension to be a positive map we require all $a_i, b_i, c_i,d_i$, $1 \leq i \leq 4$ to be non-negative. Thus $b_3+c_3=0 \implies b_3 = c_3 = 0$ but since $a_3 +c_3 =1$ we conclude $a_3=1$ in any extension. Similarly since $b_2 +c_2 = 0 \implies c_2 = 0$ but then $a_2+c_2 = \frac{1}{2} \implies a_2 =\frac{1}{2}$ in any extension. Already we have $a_2 +a_3 >1$ so it cannot be TP.

\subsection{Tensoring QSC's}
Take two QSC's,  $$\Gamma_1: SCPTP (d_1,r_1) \xrightarrow[]{} SCPTP (d_2, r_2)$$ and $$\Gamma_2 : SCPTP (d_3, r_3 ) \xrightarrow[]{} SCPTP (d_4, r_4).$$ Extend each to a superchannel $S_1$, $S_2$ respectively. Then $S_1 \otimes S_2$ is a superchannel on the combined spaces and it restricts to give a QSC on $SCPTP (d_1d_2,r_1r_2)$. This is not necessarily unique.

To see an example of this notice that for $a \neq b,$ 
$$\begin{bmatrix} 
a & 0  \\
0 & b 
\end{bmatrix} \notin S(2,1),$$ but that

$$\begin{bmatrix} 
1 & 0 & 0 & 0 \\
0 & 1 & 0 & 0 \\
0 & 0 & 1 & 0 \\
0 & 0 & 0 & 1
\end{bmatrix} \otimes \begin{bmatrix} 
a & 0  \\
0 & b 
\end{bmatrix} \in S(4,2).$$
The following two maps are extensions of the same QSC which give different outputs on $\begin{bmatrix} 
a & 0  \\
0 & b 
\end{bmatrix}$ 

$$\widetilde{S}_a \left(\begin{bmatrix} 
a & *  \\
* & b 
\end{bmatrix} \right) = a$$

$$\widetilde{S}_b \left(\begin{bmatrix} 
a & *  \\
* & b 
\end{bmatrix}\right) =b.$$
Then for any other superchannel $S$, the maps $\widetilde{S}\otimes \widetilde{S}_a$ and $\widetilde{S}\otimes \widetilde{S}_b$ are superchannels which restrict to give different QSCs on the space $S(4,2)$. But they are constructed by tensoring the same QSC's.

\section{Characterisation of superchannels}
In \cite{original} a characterisation of quantum superchannels is obtained describing them as a the action of a pre-processing channel followed by a post processing channel. It is shown that any quantum superchannel $\widetilde{\Gamma}: M_{d_1}(M_{r_1}) \xrightarrow[]{} M_{d_2}(M_{r_2})$ induces a unital CP map (i.e. the dual of a quantum channel) $\mathcal{N}:M_{d_1}\xrightarrow[]{}M_{d_2}$ which satisfies

\begin{equation}\label{1}
\Tr_{r_2}\widetilde{\Gamma}(C_{\phi}) = \mathcal{N}(\Tr_{r_1}C_{\phi})\end{equation}
where $\Tr_n$ is the partial trace tracing out system $M_n$. This $\mathcal{N}$ is where the pre-processing channel comes from.

Inserting the Kraus operators for $\widetilde{\Gamma}$ and $\mathcal{N}$ into this equation gives two different Kraus representations of the same channel. In \cite{Gour} this equation is shown to be equivalent to the Choi matrix $C_{\widetilde{\Gamma}}$ having two different purifications. Kraus representations are unique up to isometry and so are purifications (under a minimality condition). This provides a post processing channel.

This can be used to show that if $\Gamma : \mathcal{L}(M_{d_1},M_{r_1}) \xrightarrow[]{}\mathcal{L}(M_{d_2},M_{r_2})$ is a quantum superchannel there exists two quantum channels $\psi_{\text{pre}}:M_{d_2}\xrightarrow[]{}M_{d_1}\otimes M_{e}$, $\psi_{\text{post}}:M_{r_1}\otimes M_e \xrightarrow[]{}M_{r_2}$ where $e\leq d_1 d_2$ such that
$$\Gamma (\phi ) = \psi_{\text{post}}\circ (\phi \otimes \text{id}_e)\circ \psi_{\text{pre}}.$$
The dimension $e$ can be chosen to be the rank of $\Tr_{r_1}\Tr_{r_2} C_{\widetilde{\Gamma}}$ and the channel $\psi_{\text{pre}}$ can be chosen to be isometric. In \cite{2021arXiv210101552G} it is shown that this is unique in the sense that any other characterisation with equal or smaller dimension is equivalent up to action by a unitary channel.

The extension in Theorem \ref{2.7} shows that this analysis will still apply to QSC's defined on $SCPTP(d,r)$, however in this case it will be non-unique. Without extending, the characterisation cannot be derived as the Choi matrix and Kraus operators are not defined for CP maps on an operator system.

The dimension $e$ will depend on the extension of the QSC. Consider the non unique extensions from Equation (\ref{2lads}) and set $r_1=2$, the superchannel Choi matrices are
\begin{equation*}
    \begin{split}
       C_{\widetilde{\Gamma}_1}&= \begin{pmatrix}
1 & 0 & 0 & 0 \\
0 & 1 & 0 & 0 \\
0 & 0 & 0 & 0 \\
0 & 0 & 0 & 0 
\end{pmatrix}  =E_{11}+ E_{22} , \\
       C_{\widetilde{\Gamma}_2} &=  \begin{pmatrix}
0 & 0 & 0 & 0 \\
0 & 0 & 0 & 0 \\
0 & 0 & 1 & 0 \\
0 & 0 & 0 & 1 
\end{pmatrix} = E_{33}+ E_{44} .
    \end{split}
\end{equation*}
In $M_{d_1}(M_{r_1}) = M_{2}(M_2)$, we have $ \Tr_{r_1} E_{11} = E_{11} \in M_2$, $\Tr_{r_1} E_{22} = E_{11} \in M_2$, etc. Thus 
\begin{equation*}
    \begin{split}
        \Tr_{r_1}\Tr_{r_2} C_{\widetilde{\Gamma}_1} &= \begin{pmatrix}2 & 0 \\ 0 & 0 \end{pmatrix}, \\
        \Tr_{r_1}\Tr_{r_2} C_{\widetilde{\Gamma}_2} & = \begin{pmatrix}0 & 0 \\ 0 & 2 \end{pmatrix},
    \end{split}
\end{equation*}
which have equal rank. However, another equivalent extension is given by any convex combination $\widetilde{\Gamma} = p_1 \widetilde{\Gamma}_1 + p_2 \widetilde{\Gamma}_2$ for $p_1,p_2 >0$, $p_1 +p_2 =1$. This has Choi matrix
\begin{equation*}
     C_{\widetilde{\Gamma}} =  \begin{pmatrix}
p_1 & 0 & 0 & 0 \\
0 & p_1 & 0 & 0 \\
0 & 0 & p_2 & 0 \\
0 & 0 & 0 & p_2 
\end{pmatrix}, 
\end{equation*}
which reduces to \begin{equation*}
    \Tr_{r_1}\Tr_{r_2} C_{\widetilde{\Gamma}} = \begin{pmatrix} 2p_1 & 0 \\ 0 & 2p_2 \end{pmatrix},
\end{equation*}
and this has greater rank.

\begin{remark}
 The set of possible superchannel extensions of a QSC is a convex set. In the example given the extensions with minimal $e$ are extreme points. A natural question to ask is whether it is generally true that the extensions which give minimal dimensions $e$ are extreme points.
\end{remark}

Define $CP[M_n,M_m;K]$ to be CP maps from $M_n$ to $M_m$ which send the identity to a fixed $K\geq 0$. This is a convex set. The following theorem from \cite{CHOI1975285} characterises the extreme points in terms of the Kraus operators:

\begin{thm}\label{choi}

A map $\phi \in CP[M_n,M_m;K]$ is extreme if and only if it admits an expression $\phi (A) = \sum_{i}V_{i}^* AV_i$ for all $A \in M_n$ such that $\sum_i V_i^* V_i = K$ and $\{V_i^*V_j\}_{ij}$ is a linearly independent set. 
\end{thm}

Using the same proof it was noted in \cite{LANDAU1993107} that for the set of unital, trace-preserving CP maps $\phi$ is an extreme point if and only if it has Kraus operators $\{V_i\}_i$ such that $\{V_i^*V_j\bigoplus V_jV_i^*\}_{ij}$ is linearly independent.

For a CP map $\phi: M_n\to M_m$ with Kraus representation $\phi (A) = \sum_i V_i^* A V_i$ the dual map $\phi^* :M_m \to M_n$ is given by $\phi^* (B) = \sum_i V_i B V_i^*$.

\begin{defn}
Let $\mathcal{S}$ be a subspace of $M_n$, and let $\mathcal{T}$ be a subspace of $M_m$. For a CP map $\Phi:M_n \to M_m$ define the convex set $CP[M_n,M_m; \mathcal{S}, \mathcal{T}, \Phi]$ to be CP maps from $M_n$ to $M_m$ which are equal to $\Phi$ on $\mathcal{S}$ and whose duals are equal to the dual of $\Phi$ on $\mathcal{T}$.
\end{defn}

The proof of the following makes use of the same approaches as the proof of Theorem \ref{choi} from \cite{CHOI1975285}. Namely, it uses the fact that a CP map has a minimal set of Kraus operators such that they are linearly independent and any other set can be related to it via an isometry.

\begin{thm}\label{extreme}
A map $\phi \in CP[M_n,M_m; \mathcal{S}, \mathcal{T}, \Phi]$ is extreme if and only if it admits and expression $\phi (A) = \sum_{i}V_i^*AV_i$ for all $A \in M_n$ such that for any self-adjoint spanning sets $\{A_k\}_k$ for $\mathcal{S}$ and $\{B_l\}_l$ for $\mathcal{T}$ the set $$\{\bigoplus_{k} V_i^*A_kV_j \bigoplus_{l}V_j B_l V_i^*\}_{ij}$$ is linearly independent. 
\end{thm}
\begin{proof}

For the forward direction, suppose $\phi \in CP[M_n,M_m; \mathcal{S}, \mathcal{T}, \Phi]$ is extreme and take a minimal set of Kraus operators i.e. a linearly independent set $\{V_i\}_i$ with $\phi (A) = \sum_i V_i^* (A) V_i$. Choose self-adjoint spanning sets $\{A_k\}_k$ for $\mathcal{S}$ and $\{B_l\}_l$ for $\mathcal{T}$. Suppose there exist constants $\{\lambda_{ij}\}_{ij}$ such that $$\sum_{ij}\lambda_{ij}\bigoplus_{k} V_i^*A_kV_j \bigoplus_{l}V_j B_l V_i^* =0.$$

By taking the adjoint of this sum we see that $\{\overline{\lambda}_{ji}\}_{ij}$ is another set satisfying this. This implies $\{\lambda_{ij}\pm \overline{\lambda}_{ji} \}_{ij}$ do as well and if we show both these sets are the zero set then it will imply $\{\lambda_{ij}\}_{ij} = \{0\}$. Thus we may assume $(\lambda_{ij})_{ij}$ is a self-adjoint matrix. Also scale so that $-I\leq (\lambda_{ij})_{ij}\leq I$.

Define maps $\Psi_{\pm}: M_n \xrightarrow[]{}M_m$ via $$\Psi_{\pm}(A) = \sum_i V_i^* AV_i \pm \sum_{ij}\lambda_{ij}V_i^*AV_j.$$
Let $I+ (\lambda_{ij})_{ij} = (\alpha_{ij})^*(\alpha_{ij})$ so that $\sum_k \overline{\alpha}_{ki}\alpha_{kj} = \lambda_{ij}+ \delta_{ij}1$. Then if $W_i = \sum_i \alpha_{ij}V_j$ we can compute to get $\Psi_+ = \sum_i W_i^* A W_i$ and we can do similar for $\Psi_-$. This also shows that
$$\Psi_{\pm}^*(B) = \sum_{i}V_i BV_i^* \pm \sum_{ij}\lambda_{ij} V_j BV_i^*.$$
Thus $\Psi_{\pm}$ are in $CP[M_n,M_m; \mathcal{S}, \mathcal{T}, \phi]$.

We now have $\phi = \frac{1}{2}(\Psi_+ + \Psi_-)$ and so since it is extreme $\phi = \Psi_+$. The minimality of the set $\{V_i\}$ implies that $(\alpha_{ij})_{ij}$ is an isometry which gives $(\lambda_{ij})_{ij} = 0$ and we are done.

Conversely, assume $\phi$ has form $\phi (A) = \sum_i V_i^* AV_i$ and $$\{ \bigoplus_{k} V_i^*A_kV_j \bigoplus_{l}V_j B_l V_i^*\}_{ij}$$ is linearly independent for any self-adjoint spanning sets $\{A_k\}_k$ for $\mathcal{S}$ and $\{B_l\}_l$ for $\mathcal{T}$. This implies $\{V_i\}_i$ is linearly independent since $\sum_i \lambda_i V_i = 0$ would imply for any arbitrary summand that $\sum_{ij}\lambda_i V_j^*CV_i = 0$.

If $\phi$ is not extreme, say $\phi = \frac{1}{2}(\Psi_1 +\Psi_2)$ for $\Psi_1 (A) = \sum_p W_p^* A W_p$ and $\Psi_2 (A) = \sum_q Z_q^* A Z_q$, then we can write $W_p$ and $Z_q$ in terms of $V_i$. But if $W_p =\sum_{i}\alpha_{pi} V_i $ we have

$$\sum_i V_i^*A_kV_i = \sum_p W_p^* A_k W_p = \sum_{pij}\overline{\alpha}_{pi}\alpha_{pj}V_i^* A_k V_j$$
for any $A_k$ in the spanning set. Similarly for the dual we have $\sum_{i}V_i B_k V_i^* = \sum_{pij}\alpha_{pj}\overline{\alpha}_{pi}V_jB_k V_i^* $. Therefore $\sum_{p}\overline{\alpha}_{pi}\alpha_{pj} = \delta_{ij}$ or else we would have a linear dependency. This implies $(\alpha_{pi})_{pi}$ is an isometry and so the Kraus operators it relates define the same map; i.e., $\phi = \Psi_1$ and so $\phi$ is extreme.
\end{proof}

\begin{remark}
This immediately applies to the set of superchannels $\widetilde{S}: M_{d_1}(M_{r_1}) \xrightarrow[]{} M_{d_2}(M_{r_2})$ which are extensions of the same QSC; i.e., are equal on the space $\mathcal{S} = S(d_1,r_1)$. A spanning set consisting of Choi matrices of quantum channels may be chosen. In this case the space $\mathcal{T}$ may be chosen to be zero. For a trace-preserving superchannel we can consider the set of TP extensions of the underlying QSC, these are the ones with $\mathcal{T}=\text{span}\{I_{d_2r_2}\}$ being sent to $\text{span}\{I_{d_1r_1}\}$ (since $\phi$ being trace-preserving is equivalent to $\phi^*$ being unital). As shown for some QSCs this set of extensions is empty.
\end{remark}

\begin{exam}[Unitary superchannels]
If $U_1 \in M_{d_1}$ and $U_2 \in M_{r_1}$ are unitaries then it's easy to see conjugation by $U_1 \otimes U_2$ is a superchannel since if $\phi$ is a TP map, then
\begin{equation*}
    \begin{split}
        \Tr_{r_2}[U_1 \otimes U_2 (C_{\phi})(U_1 \otimes U_2)^*] &= \Tr_{r_2}[U_1 \otimes U_2 (\sum_{ij}E_{ij} \otimes \phi (E_{ij}))U_1^* \otimes U_2^*] \\
        &=\sum_{ij}(U_1 E_{ij}U_1^*)\cdot \Tr(U_2 \phi (E_{ij})U_2^*) \\
        &= U_1 (\sum_{i}E_{ii})U_1^* = I_{d_2},
    \end{split}
\end{equation*}
so it satisfies the TP-preserving condition.
\end{exam}
In fact every unitary superchannel is of this form. Let $\mathcal{U}(n)$ denote the unitaries in $M_n$.
\begin{thm}
If $U \in \mathcal{U}(dr)$ is a unitary such that the map $\widetilde{S}:M_d (M_r) \xrightarrow[]{}M_d (M_r)$ with $\widetilde{S}(C) = UCU^*$ is a superchannel then there exists unitaries $U_1 \in \mathcal{U}(d)$ and $U_2 \in \mathcal{U}(r)$ such that $U=U_1 \otimes U_2$.
\end{thm}
\begin{proof}
In \cite[Lemma 2]{original} it is shown that the dual of a superchannel must satisfy
$$\widetilde{S}^* (\rho \otimes I_{r_2}) = N^* (\rho) \otimes I_{r_1}, \quad \rho \in M_{d_2}$$
for some quantum channel $N^* : M_{d_2}\xrightarrow[]{}M_{d_1}$. Now for any $X, Y \in M_d$ and $Z \in M_{r}$ with $\Tr Z = 0$ we have
\begin{equation*}
    \begin{split}
        \Tr [U(X\otimes Z)U^* (Y \otimes I_r)] = \Tr[(X\otimes Z) (N^* (Y) \otimes I_{r})] = 0.
    \end{split}
\end{equation*}
Now by \cite[Theorem 2.3]{2016JPhA...49G5301D} we are done.
\end{proof}
By Theorem \ref{extreme} any unitary superchannel $\widetilde{S}$ is an extreme point of the set of extensions of the underlying QSC. They will also always have minimal dimension $e$ for the characterisation theorem since the rank of $\Tr_{r_1}\Tr_{r_2}C_{\widetilde{S}}$ will be 1. To see this write the matrix units of $M_d(M_r)$ as $E_{ij}\otimes F_{kl}$, $1\leq i,j\leq d$, $1\leq k,l\leq r$, where $E_{ij} \in M_d$ and $F_{kl} \in M_r$ are the standard matrix units in their spaces. 
Then since $S(C) = U_1\otimes U_2 C U_1^* \otimes U_2^*$ for $U_1 \in \mathcal{U}(d)$, $U_2\in \mathcal{U}(r)$ we have
$$C_{\widetilde{S}} = \sum_{i,j}\sum_{k,l} E_{i,j}\otimes F_{k,l} \otimes U_1E_{i,j}U_1^* \otimes U_2 F_{k,l}U_2^*.$$
Now applying $\Tr_{r_1}\Tr_{r_2}$ traces out the 2nd and 4th term giving
\begin{equation*}
\begin{split}
\Tr_{r_1}\Tr_{r_2}C_{\widetilde{S}} &=r\cdot \sum_{i,j} E_{i,j}\otimes U_1 E_{i,j}U_1^* \\
&= \text{Diag}(U_1,\ldots, U_1) \sum_{i,j}E_{i,j} \otimes E_{i,j}\text{Diag}(U_1^*,\ldots, U_1^*).
\end{split}
\end{equation*}
Since $\text{Diag}(U_1,\ldots, U_1)$ has full rank and $\sum_{i,j}E_{i,j}\otimes E_{i,j}$ has rank 1 the overall matrix has rank 1.

\section{Conclusion}
Our results show that defining superchannels to act on the space of quantum channels gives a different class of maps in comparison to the original definition of superchannels. The standard definition of superchannel can be recovered by extending to the full set of linear maps and this extension is not unique. Therefore many different quantum superchannels can restrict to the same QSC, which means they are effectively the same as maps on channels. Thus, if we are really only concerned with the action of a superchannel on quantum channels, then we are really only concerned with the corresponding QSC.

It would be interesting to see if there is a ``best" choice of extension. For example it may be that the minimal dimension $e$ for superchannel characterisation occurs for extreme points of the set of extensions, but this was not proved here and is an open question. It was shown that TP extensions are not always available. It is unclear what the restrictions are on the choice of extension.

Not much is known about the operator system of quantum channels. It might be worth considering how the action of a map on this space determines the form of its possible Stinespring representations, and whether this affects the characterisation of superchannels.

\subsubsection*{Acknowledgements} The author would like to thank Vern Paulsen, David Kribs, and Michael Brannan for helpful comments and discussion. This work was completed as part of the authors PhD thesis.

\nocite{*}
\bibliographystyle{amsplain}
\bibliography{refer}

\end{document}